\documentclass[%
 reprint,
 amsmath,amssymb,
 aps,
]{revtex4-1}

\usepackage{amsthm}
\usepackage{graphicx}
\usepackage{dcolumn}
\usepackage{bm}

\begin{document}

\preprint{APS/123-QED}

\title{Local indistinguishability of orthogonal product states}

\author{Zhi-Chao Zhang$^{1}$}
\author{Fei Gao$^{1}$}
 \email{gaofei\_bupt@hotmail.com}
\author{Ya Cao$^{1}$}%
\author{Su-Juan Qin$^{1}$}%
\author{Qiao-Yan Wen$^{1}$}

\affiliation{%
 $^{1}$State Key Laboratory of Networking and Switching Technology, Beijing University of Posts and Telecommunications, Beijing, 100876, China
}%

\date{\today}

\begin{abstract}
In the general bipartite quantum system $m \otimes n$, Wang \emph{et al.} [Y.-L Wang \emph{et al.}, Phys. Rev. A \textbf{92}, 032313 (2015)] presented $3(m+n)-9$ orthogonal product states which cannot be distinguished by local operations and classical communication (LOCC). In this paper, we aim to construct less locally indistinguishable orthogonal product states in $m\otimes n $. First, in $3\otimes n (3< n)$ quantum system, we construct $3n-2$ locally indistinguishable orthogonal product states which are not unextendible product bases. Then, for $m\otimes n (4\leq m\leq n)$, we present $3n+m-4$ orthogonal product states which cannot be perfectly distinguished by LOCC. Finally, in the general bipartite quantum system $m\otimes n(3\leq m\leq n)$,  we show a smaller set with $2n-1$ orthogonal product states and prove that these states are LOCC indistinguishable using a very simple but quite effective method. All of the above results demonstrate the phenomenon of nonlocality without entanglement.

\begin{description}
\item[PACS numbers]
03.67.Hk, 03.65.Ud
\end{description}
\end{abstract}

\pacs{Valid PACS appear here}
\maketitle


\section{\label{sec:level1}Introduction\protect}
In quantum information theory, one of the main goals is to understand the power and limitation of quantum operations which can be implemented by local operations and classical communication (LOCC). When global operators cannot be implemented by LOCC, it reflects the fundamental feature of quantum mechanics which is called nonlocality and has received
wide attention in recent years [1-5]. Especially for the phenomenon of nonlocality without entanglement, many interesting results have been presented [6-14].

In general, it is well known that entanglement increases the difficulty of distinguishing orthogonal quantum states by LOCC [15-22]. However, many results reveal that entanglement is not necessary for LOCC indistinguishable quantum states. In 1999, Bennett \emph{et al.} presented nine LOCC indistinguishable pure product states in $3\otimes3$ and showed the phenomenon of nonlocality without entanglement, which is a fundamental result in this area [6]. Furthermore, for the nonlocality of the nine pure product states, Walgate \emph{et al.} put forward a very simple proof method [7]. Then, this result was generalized in $d\otimes d$ [23,24].
In [25], Childs \emph{et al.} proved that any LOCC measurement for discriminating irreducible domino-type tilings would err with certain probability.

Despite these huge advances, the nonlocality of orthogonal product states is still extensively studied. Recently, Wang \emph{et al.} presented $3(m+n)-9$ locally indistinguishable orthogonal product states in $m \otimes n$[26]. In general, in the problems of distinguishability with LOCC what matters is the minimum number of LOCC indistinguishable states. Naturally, for the general $m\otimes n(3\leq m\leq n)$ quantum system, it is interesting to ask whether there exist less locally indistinguishable orthogonal product states. Therefore, finding such states is still meaningful.

In this paper, we focus on constructing the locally indistinguishable orthogonal product states in the general bipartite quantum systems. Fortunately, we construct $3n-2$ orthogonal product states in $3\otimes n (3< n)$ and $3n+m-4$ orthogonal product states in $m\otimes n (4\leq m\leq n)$. Then, we prove these states cannot be perfectly distinguished by LOCC using a very simple but quite effective method. And we show that these states are not unextendible product bases [12] but other classes of locally indistinguishable orthogonal product states. To better understand the phenomenon of nonlocality
without entanglement, in $m\otimes n (3\leq m\leq n)$, we present a smaller set with $2n-1$ locally indistinguishable orthogonal product states. And most of these states are extendible. In addition, the two classes of orthogonal product states of our construction are less than the states in [26].

The rest of this paper is organized as follows. In Sec. \uppercase\expandafter{\romannumeral2}, we present $3n-2$ locally indistinguishable orthogonal product states in $3\otimes n (3< n)$ and $3n+m-4$ locally indistinguishable orthogonal product states in $m\otimes n (4\leq m\leq n)$. In Sec. \uppercase\expandafter{\romannumeral3}, we construct $2n-1$ orthogonal product states in $m\otimes n (3\leq m\leq n)$, which cannot be perfectly distinguished by LOCC. Finally, in Sec. \uppercase\expandafter{\romannumeral4}, we draw the conclusion.

\theoremstyle{remark}
\newtheorem{definition}{\indent Definition}
\newtheorem{lemma}{\indent Lemma}
\newtheorem{theorem}{\indent Theorem}
\newtheorem{corollary}{\indent Corollary}

\def\QEDclosed{\mbox{\rule[0pt]{1.3ex}{1.3ex}}}
\def\QED{\QEDclosed}
\def\proof{\indent{\em Proof}.}
\def\endproof{\hspace*{\fill}~\QED\par\endtrivlist\unskip}

\section{Local Indistinguishability of Orthogonal Product States}

In this section, we will construct the orthogonal product states in the general bipartite quantum system and prove that these states cannot be perfectly distinguished by LOCC.

When $m=n=3$, there are 8 locally indistinguishable orthogonal product states [24]. And when we add $|\phi_{9}\rangle=|0\rangle_{A}|0\rangle_{B}$ in the set of states, these quantum states are in fact the nine states which were presented by Bennett \emph{et al.} in [6].

In the following, we present the main results.

Firstly, in $3\otimes n (3< n)$, we construct $3n-2$ orthogonal product states as follows:

\begin{eqnarray}
\label{eq.2}
\begin{split}
&|e\pm f\rangle=\frac{1}{\sqrt{2}}(|e\rangle \pm |f\rangle),0\leq e<f\leq n-1,\\
&|\phi_{i}\rangle=|i\rangle_{A}|0-i\rangle_{B}, i=1,2,\\ &|\phi_{i+2}\rangle=|0-i\rangle_{A}|j\rangle_{B},i=1,j=2;i=2,j=1,\\
&|\phi_{i+4}\rangle=|i\rangle_{A}|0+i\rangle_{B}, i=1,2,\\ &|\phi_{7}\rangle=|0+2\rangle_{A}|1\rangle_{B},\\
&|\phi_{j+5}\rangle=|0-1\rangle_{A}|j\rangle_{B},j=3,\ldots,n-1,\\
&|\phi_{i+2(r-1)+n+4}\rangle=|j\rangle_{A}|(r+i)-(n-r)\rangle_{B},\\
&i=1,|j\rangle=|0+1\rangle;i=2,j=2,r=1,\ldots,a-1,\\
&|\phi_{i+2(a-1)+n+4}\rangle=|j\rangle_{A}|(a+i)-(n-a)\rangle_{B},\\
&i=b=1,|j\rangle=|0+1\rangle,a \geq 2,\\
&|\phi_{i+2(r-1)+2n+1}\rangle=|j\rangle_{A}|(r+i)+(n-r)\rangle_{B},\\
&i=1,|j\rangle=|0+1\rangle;i=2,j=2,r=1,\ldots,a-1,\\
&|\phi_{i+2(a-1)+2n+1}\rangle=|j\rangle_{A}|(a+i)+(n-a)\rangle_{B},\\
&i=b=1,|j\rangle=|0+1\rangle,a \geq 2.\\
\end{split}
\end{eqnarray}

where $n=2a+b+1, a\geq1, 0\leq b<2$.

In the following, we prove that these quantum states (1) are LOCC indistinguishable.

\begin{theorem}
In $3\otimes n (3< n)$,  the above $3n-2$ states
cannot be perfectly
distinguished by LOCC.
\end{theorem}

\begin{proof}
To distinguish the states, some party has to start with a \emph{nontrivial} and \emph{non-disturbing} measurement, i.e., not all measurements $M_{m}^{\dagger}M_{m}$ are proportional to the identity and have the orthogonality relations preserved afterwards, making further discrimination possible [27].

First, Alice performs a general measurement which is represented by
a set of general $3\times 3$ positive
operator-valued measure (POVM) elements $M_{3}^{\dagger}M_{3}$. In the $\{|0\rangle,|1\rangle,|2\rangle\}_{A}$ basis which corresponds to the states (1), we
write the POVM elements

$M_{3}^{\dagger}M_{3}=\left[
  \begin{array}{ccccc}
    a_{00} & a_{01} &   a_{02} \\
    a_{10} & a_{11} &   a_{12} \\
    a_{20} & a_{21} &   a_{22} \\
  \end{array}
\right]
$.

The post measurement states $\{M_{3}\otimes I_{B}|\phi_{j}\rangle, j=1,\ldots,3n-2\}$ should also be mutually orthogonal. Considering the states $|\phi_{1,2}\rangle$, we know $\langle 1|M_{3}^{\dagger}M_{3}|2\rangle\langle 0-1|0-2\rangle=0$. Thus, $\langle 1|M_{3}^{\dagger}M_{3}|2\rangle=0, i.e., a_{12}=a_{21}=0 $.

For the states $|\phi_{j}\rangle$ and $|\phi_{i+2}\rangle, i=1$, $j=2$, and $i=2, j=1$, we have $\langle j|M_{3}^{\dagger}M_{3}|0-i\rangle\langle 0-j|j\rangle=0$. Then, $\langle j|M_{3}^{\dagger}M_{3}|0-i\rangle=\langle j|M_{3}^{\dagger}M_{3}|0\rangle=0, i.e., a_{0j}=a_{j0}=0$, $ j=1,2$.

Lastly, considering the states $|\phi_{i+2}\rangle, i=1,2$, $|\phi_{7}\rangle$ and $|\phi_{n+5}\rangle$, we know $\langle0+i|M_{3}^{\dagger}M_{3}|0-i\rangle\langle j|j\rangle=0, i.e., \langle 0|M_{3}^{\dagger}M_{3}|0\rangle-\langle i|M_{3}^{\dagger}M_{3}|i\rangle=0, i=1,2$. Then, $\langle 0|M_{3}^{\dagger}M_{3}|0\rangle=\langle i|M_{3}^{\dagger}M_{3}|i\rangle, i=1,2$. Thus, $a_{00}=a_{11}=a_{22}$.

Therefore, all of Alice's measurements $M_{3}^{\dagger}M_{3}$
are proportional to the identity, meaning that Alice cannot start with a nontrivial measurement.

When Bob has to start with a nontrivial and non-disturbing measurement $M_{n}^{\dagger}M_{n}$, we
write the POVM elements in the $\{|0\rangle,|1\rangle,\ldots,|n-2\rangle,|n-1\rangle\}_{A}$ basis which corresponds to the states (1),

$M_{n}^{\dagger}M_{n}=\left[
  \begin{array}{ccccc}
    a_{00} & a_{01} &  \cdots & a_{0n-1} \\
    a_{10} & a_{11} &  \cdots & a_{1n-1} \\
    \vdots & \vdots &  \ddots & \vdots \\

    a_{n-10} & a_{n-11} &  \cdots & a_{n-1n-1} \\
  \end{array}
\right]
$.

Then, we can also get that the post measurement states $\{I_{A}\otimes M_{n}|\phi_{i}\rangle, i=1,\ldots,3n-2\}$ should be mutually orthogonal. In the same way, considering the states $|\phi_{i+2}\rangle, i=1,2$ and $|\phi_{j+5}\rangle, j=3,\ldots,n-1$, we get $\langle i|M_{n}^{\dagger}M_{n}|j\rangle=0, i.e., a_{ij}=0, i,j=1,\ldots,n-1, i\neq j$.

For the states $|\phi_{i}\rangle$ and $|\phi_{i+2}\rangle$, $i=1,2$, we have $\langle i|0-i\rangle\langle 0-i|M_{n}^{\dagger}M_{n}|j\rangle=0$. Then, $\langle 0-i|M_{n}^{\dagger}M_{n}|j\rangle=\langle 0|M_{n}^{\dagger}M_{n}|j\rangle=0, i.e., a_{0j}=a_{j0}=0, j=1,2$.

For the states $|\phi_{1,2}\rangle$ and $|\phi_{i+2(r-1)+n+4}\rangle$, where $i=1,2, r=1,\ldots,a-1$, and $i=b=1, r=a$, we have $\langle 0-i|M_{n}^{\dagger}M_{n}|(r+i)-(n-r)\rangle=0$. Then, $\langle 0|M_{n}^{\dagger}M_{n}|(n-r)\rangle=\langle 0|M_{n}^{\dagger}M_{n}|(r+i)\rangle=\langle 0|M_{n}^{\dagger}M_{n}|2\rangle=0$, $i.e., a_{0n-r}=a_{n-r0}=a_{0r+i}=a_{r+i0}=a_{20}=0$, where $i=1,\ldots,m-1, r=1,\ldots,a-1$, and $i=1,\ldots,b, r=a$.

Therefore, we know $a_{0i}=a_{i0}=0, i=1,\ldots,n-1$.

For the states $|\phi_{i}\rangle$ and $|\phi_{i+4}\rangle$, $i=1,2$, we know $\langle i|i\rangle\langle 0+i|M_{n}^{\dagger}M_{n}|0-i\rangle=0, i.e., \langle 0|M_{n}^{\dagger}M_{n}|0\rangle=\langle i|M_{n}^{\dagger}M_{n}|i\rangle$, $i=1,2$. In the same way, for the states $|\phi_{i+2(r-1)+n+4}\rangle$ and $|\phi_{i+2(r-1)+2n+1}\rangle$, where $i=1,2, r=1,\ldots,a-1$, and $i=1=b, r=a$, we can also get $\langle i|M_{n}^{\dagger}M_{n}|i\rangle=\langle 2|M_{n}^{\dagger}M_{n}|2\rangle$, $i=3,\ldots,n-1$. Thus, $a_{00}=a_{11}=\cdots=a_{n-1n-1}$. That is, all of Bob's measurements $M_{n}^{\dagger}M_{n}$
are proportional to the identity.
Thus, Bob cannot start with a nontrivial measurement either. Therefore,
the $3n-2$ states
cannot be perfectly
distinguished by LOCC. This completes
the proof.
\end{proof}

In addition, the states of our construction are less than the states in [26] and are not unextendible product bases [12], either. For example, the following product state $|\phi_{3n+m-3}\rangle=|0\rangle_{A}|0\rangle_{B}$ is orthogonal to the states in (1).

In the higher-dimensional general bipartite quantum system $m\otimes n (4\leq m\leq n)$, we present the following $3n+m-4$ orthogonal product states:

\begin{eqnarray}
\label{eq.1}
\begin{split}
&|e\pm f\rangle=\frac{1}{\sqrt{2}}(|e\rangle \pm |f\rangle),0\leq e<f\leq n-1,\\
&|\phi_{i}\rangle=|i\rangle_{A}|0-i\rangle_{B}, i=1,\ldots,m-1,\\ &|\phi_{i+m-1}\rangle=|0-i\rangle_{A}|j\rangle_{B},\\
&i=1,\ldots,m-2,j=i+1; i=m-1,j=1,\\
&|\phi_{i+2m-2}\rangle=|i\rangle_{A}|0+i\rangle_{B},i=1,\ldots,m-1,\\
&|\phi_{i+3m-3}\rangle=|0+i\rangle_{A}|j\rangle_{B},\\
&i=1,\ldots,m-2,j=i+1;i=m-1,j=1,\\
&|\phi_{j+3m-3}\rangle=|0\rangle_{A}|j\rangle_{B},j=m,\ldots,n-1,\\
&|\phi_{i+(r-1)(m-1)+n+3m-4}\rangle=|i\rangle_{A}|[r(m-2)+i]-(n-r)\rangle_{B},\\
&i=1,\ldots,m-1,r=1,\ldots,a-1,\\
&|\phi_{i+(a-1)(m-1)+n+3m-4}\rangle=|i\rangle_{A}|[a(m-2)+i]-(n-a)\rangle_{B},\\
&i=1,\ldots,b,\\
&|\phi_{i+(r-1)(m-1)+2n+2m-4}\rangle=|i\rangle_{A}|[r(m-2)+i]+(n-r)\rangle_{B},\\
&i=1,\ldots,m-1,r=1,\ldots,a-1,\\
&|\phi_{i+(a-1)(m-1)+2n+2m-4}\rangle=|i\rangle_{A}|[a(m-2)+i]+(n-a)\rangle_{B},\\
&i=1,\ldots,b.\\
\end{split}
\end{eqnarray}

where $n=a(m-1)+b+1, a\geq1, 0\leq b< m-1.$

In the following, we study the local indistinguishability of these quantum states (2).

\begin{theorem}
In $m\otimes n (4\leq m\leq n)$,  the above $3n+m-4$ states
cannot be perfectly
distinguished by LOCC.
\end{theorem}

\begin{proof}
Similarly to the proof of Theorem 1, we first need to prove that Bob cannot start with a nontrivial measurement.

When Bob has to start with a nontrivial and non-disturbing measurement $M_{n}^{\dagger}M_{n}$, we
write the POVM elements in the $\{|0\rangle,|1\rangle,\ldots,|n-2\rangle,|n-1\rangle\}_{A}$ basis which corresponds to the states (1),

$M_{n}^{\dagger}M_{n}=\left[
  \begin{array}{ccccc}
    a_{00} & a_{01} &  \cdots & a_{0n-1} \\
    a_{10} & a_{11} &  \cdots & a_{1n-1} \\
    \vdots & \vdots &  \ddots & \vdots \\

    a_{n-10} & a_{n-11} &  \cdots & a_{n-1n-1} \\
  \end{array}
\right]
$.

Then, we can also get that the post measurement states $\{I_{A}\otimes M_{n}|\phi_{i}\rangle, i=1,\ldots,3n+m-4\}$ should be mutually orthogonal. In the same way, considering the states $|\phi_{i+m-1}\rangle, i=1,\ldots,m-2,j=i+1; i=m-1,j=1$ and $|\phi_{j+3m-3}\rangle, j=m,\ldots,n-1$, we get $\langle i|M_{n}^{\dagger}M_{n}|j\rangle=0, i.e., a_{ij}=0, i,j=1,\ldots,n-1, i\neq j$.

For the states $|\phi_{i}\rangle$ and $|\phi_{i+m-1}\rangle$, $i=1,\ldots,m-1$, we have $\langle i|0-i\rangle\langle 0-i|M_{n}^{\dagger}M_{n}|j\rangle=0$. Then, $\langle 0-i|M_{n}^{\dagger}M_{n}|j\rangle=\langle 0|M_{n}^{\dagger}M_{n}|j\rangle=0, i.e., a_{0j}=a_{j0}=0, j=1,\ldots,m-1$.

For the states $|\phi_{i}\rangle$ and $|\phi_{i+(r-1)(m-1)+n+3m-4}\rangle$, where $i=1,\ldots,m-1, r=1,\ldots,a-1$, and $i=1,\ldots,b, r=a$, we have $\langle 0-i|M_{n}^{\dagger}M_{n}|[r(m-2)+i]-(n-r)\rangle=0$. Then, $\langle 0|M_{n}^{\dagger}M_{n}|(n-r)\rangle=\langle 0|M_{n}^{\dagger}M_{n}|[r(m-2)+i]\rangle=\langle 0|M_{n}^{\dagger}M_{n}|m-1\rangle=0$, $i.e., a_{0n-r}=a_{n-r0}=a_{0r(m-2)+i}=a_{r(m-2)+i0}=a_{m-10}=0$, where $i=1,\ldots,m-1, r=1,\ldots,a-1$, and $i=1,\ldots,b, r=a$.

Therefore, we know $a_{0i}=a_{i0}=0, i=1,\ldots,n-1$.

For the states $|\phi_{i}\rangle$ and $|\phi_{i+2m-2}\rangle$, $i=1,\ldots,m-1$, we know $\langle i|i\rangle\langle 0+i|M_{n}^{\dagger}M_{n}|0-i\rangle=0, i.e., \langle 0|M_{n}^{\dagger}M_{n}|0\rangle=\langle i|M_{n}^{\dagger}M_{n}|i\rangle$, $i=1,\ldots,m-1$. In the same way, for the states $|\phi_{i+(r-1)(m-1)+n+3m-4}\rangle$ and $|\phi_{i+(r-1)(m-1)+2n+2m-4}\rangle$, where $i=1,\ldots,m-1, r=1,\ldots,a-1$, and $i=1,\ldots,b, r=a$, we can also get $\langle i|M_{n}^{\dagger}M_{n}|i\rangle=\langle m-1|M_{n}^{\dagger}M_{n}|m-1\rangle$, $i=m,\ldots,n-1$. Thus, $a_{00}=a_{11}=\cdots=a_{n-1n-1}$. That is, all of Bob's measurements $M_{n}^{\dagger}M_{n}$
are proportional to the identity.
Thus, Bob cannot start with a nontrivial measurement.

When Alice has to start with the nontrivial and non-disturbing measurements $M_{m}^{\dagger}M_{m}$, we can also get that the post measurement states $\{M_{m}\otimes I_{B}|\phi_{i}\rangle, i=1,\ldots,3n+m-4\}$ should be mutually orthogonal. In the same way, all of Alice's measurements $M_{m}^{\dagger}M_{m}$
are also proportional to the identity.
That is, Alice cannot start with a nontrivial measurement either.
Therefore,
the $3n+m-4$ states
cannot be perfectly
distinguished by LOCC. This completes
the proof.
\end{proof}

In fact, when $m=n$, the states (2) can be extended to an orthogonal product base by adding the following $m^{2}-4m+4$ product states $\{|\phi_{00}\rangle=|0\rangle_{A}|0\rangle_{B}$,$|\phi_{ij}\rangle=|i\rangle_{A}|j\rangle_{B}, i=1,\ldots,m-2,j=1,\ldots,m-1,j\neq i, j\neq i+1$ and $i=m-1, j=2,\ldots,m-2\}$. In the following, we show the LOCC indistinguishable product states can be perfectly distinguished by separable measurements. First, we denote these states as $|\phi_{i}\rangle, i=1,\ldots, m^{2}$. Then, we define a measurement $\{M_{i}\}_{i=1}^{m^{2}}$, where $M_{i}=|\phi_{i}\rangle\langle\phi_{i}|$. Because $\{|\phi_{i}\rangle\}_{i=1}^{m^{2}}$ is an orthogonal product base in $m\otimes m$. Thus, $\sum_{i=1}^{m^{2}}M_{i}=\sum_{i=1}^{m^{2}}|\phi_{i}\rangle\langle\phi_{i}|=I$. As $|\phi_{i}\rangle$ is a product state, we can get $M_{i}$ is separable. Therefore, the set of states can be perfectly distinguished by separable measurement $\{M_{i}\}_{i=1}^{m^{2}}$. Naturally, when $m=n$, the states (2) can also be perfectly distinguished by separable measurement. This result also means that separable operations are strictly stronger than the local operations and classical communication.

\section{Less Locally Indistinguishable Orthogonal Product States}

In [27], Yu \emph{et al.} presented $2d-1$ LOCC indistinguishable orthogonal product states in $d\otimes d(d>2)$. In this section, we generalize the states to arbitrary bipartite quantum systems $m\otimes n$ and present a very simple but quite effective proof.

\begin{theorem}
In $m\otimes n(3\leq m\leq n)$, there are $2n-1$ LOCC indistinguishable orthogonal product states.
\end{theorem}

\begin{proof}
First, we construct $2n-1$ orthogonal product states as follows:

\begin{eqnarray}
\label{eq.2}
\begin{split}
&|e\pm f\rangle=\frac{1}{\sqrt{2}}(|e\rangle \pm |f\rangle),0\leq e<f\leq n-1,\\
&|\phi_{i}\rangle=|i\rangle_{A}|0-i\rangle_{B}, i=1,\ldots,m-1,\\ &|\phi_{i+m-1}\rangle=|0-i\rangle_{A}|j\rangle_{B},\\
&i=1,\ldots,m-2,j=i+1; i=m-1,j=1,\\
&|\phi_{j+m-1}\rangle=|0-1\rangle_{A}|j\rangle_{B},j=m,\ldots,n-1,\\
&|\phi_{n+m-1}\rangle=|0+1\rangle_{A}|2-(n-1)\rangle_{B}, m<n\\
&|\phi_{i+n+m-2}\rangle=|i\rangle_{A}|(m-2+i)-(n-1)\rangle_{B},\\
&i=2,\ldots,m-1,\\
&|\phi_{i+(r-1)(m-1)+n+m-2}\rangle=|j\rangle_{A}|[r(m-2)+i]-(n-r)\rangle_{B},\\
&i=1,|j\rangle=|0+1\rangle;i=2,\ldots,m-1,j=i,r=2,\ldots,a-1,\\
&|\phi_{i+(a-1)(m-1)+n+m-2}\rangle=|j\rangle_{A}|[a(m-2)+i]-(n-a)\rangle_{B},\\
&i=1,|j\rangle=|0+1\rangle;i=2,\ldots,b,j=i,a \geq 2,\\
&|\phi_{2n-1}\rangle=|0+1+\cdots+m-1\rangle_{A}|0+1+\cdots+n-1\rangle_{B}.\\
\end{split}
\end{eqnarray}

where $n=a(m-1)+b+1, a\geq1, 0\leq b<m-1$.

In the following, we will prove that these states (3) are locally indistinguishable. Similarly to the proof of Theorem 1, we first need to prove that Alice cannot start with a nontrivial measurement.

When Alice has to start with the nontrivial and nondisturbing  measurement $M_{m}^{\dagger}M_{m}$, we write measurement $M_{m}^{\dagger}M_{m}$ in the $\{|0\rangle,|1\rangle,\ldots,|m-2\rangle,|m-1\rangle\}_{A}$ basis which corresponds to the states (3):

$M_{m}^{\dagger}M_{m}=\left[
  \begin{array}{ccccc}
    a_{00} & a_{01} &  \cdots & a_{0m-1} \\
    a_{10} & a_{11} &  \cdots & a_{1m-1} \\
    \vdots & \vdots &  \ddots & \vdots \\

    a_{m-10} & a_{m-11} &  \cdots & a_{m-1m-1} \\
  \end{array}
\right]
$.

The post measurement states $\{M_{m}\otimes I_{B}|\phi_{j}\rangle, j=1,\ldots,2n-1\}$ should also be mutually orthogonal. Considering the states $|\phi_{i}\rangle, i=1,\ldots,m-1$, we know $\langle i|M_{m}^{\dagger}M_{m}|j\rangle\langle 0-i|0-j\rangle=0, i,j=1,\ldots,m-1, i\neq j$. Thus, $\langle i|M_{m}^{\dagger}M_{m}|j\rangle=0, i.e., a_{ij}=0, i,j=1,\ldots,m-1, i\neq j$.

For the states $|\phi_{j}\rangle$ and $|\phi_{i+m-1}\rangle, i=1,\ldots,m-2$, $j=i+1$, and only for $i=m-1, j=1$, we have $\langle j|M_{m}^{\dagger}M_{m}|0-i\rangle\langle 0-j|j\rangle=0$. Then, $\langle j|M_{m}^{\dagger}M_{m}|0-i\rangle=\langle j|M_{m}^{\dagger}M_{m}|0\rangle=0, i.e., a_{0j}=a_{j0}=0, j=1,\ldots,m-1$.

Lastly, considering the states $|\phi_{i+m-1}\rangle$ and $|\phi_{2n-1}\rangle$, $i=1,\ldots,m-2$, $j=i+1$, and only for $i=m-1, j=1$, we know $\langle 0+i|M_{m}^{\dagger}M_{m}|0-i\rangle\langle j|j\rangle=0, i.e., \langle 0|M_{m}^{\dagger}M_{m}|0\rangle-\langle i|M_{m}^{\dagger}M_{m}|i\rangle=0, i=1,\ldots,m-1$. Then, $\langle 0|M_{m}^{\dagger}M_{m}|0\rangle=\langle i|M_{m}^{\dagger}M_{m}|i\rangle, i=1,\ldots,m-1$. Thus, $a_{00}=a_{11}=\cdots=a_{m-1m-1}$.

Therefore, all measurements $M_{m}^{\dagger}M_{m}$ are proportional to the identity. That is to say, Alice cannot start with a nontrivial measurement.

When Bob has to start with the nontrivial and non-disturbing measurements $M_{n}^{\dagger}M_{n}$, we can also get that the post measurement states $\{I_{A}\otimes M_{n}|\phi_{i}\rangle, i=1,\ldots,2n-1\}$ should be mutually orthogonal. In the same way, all of Bob's measurements $M_{n}^{\dagger}M_{n}$
are also proportional to the identity.
That is, Bob cannot start with a nontrivial measurement either.
Therefore, the states (3) cannot be perfectly
distinguished by LOCC.
This completes
the proof.
\end{proof}

From the result, we can see the number of the states (3) is much less than the states in [26]. For example, in $3\otimes4$, the states of our construction has only 7 states, but they presented 12 states. In addition, most of our set are extendible. When $m=n=3$, the states (3) are unextendible. However, in $3\otimes5$, from (3), we can get the following states :
\begin{eqnarray}
\label{eq.3}
\begin{split}
&|e\pm f\rangle=\frac{1}{\sqrt{2}}(|e\rangle \pm |f\rangle),0\leq e<f\leq 4,\\
&|\phi_{1}\rangle=|1\rangle_{A}|0-1\rangle_{B},\\ &|\phi_{2}\rangle=|2\rangle_{A}|0-2\rangle_{B},\\
&|\phi_{3}\rangle=|0-1\rangle_{A}|2\rangle_{B},\\
&|\phi_{4}\rangle=|0-2\rangle_{A}|1\rangle_{B},\\
&|\phi_{5}\rangle=|0-1\rangle_{A}|3\rangle_{B},\\
&|\phi_{6}\rangle=|0-1\rangle_{A}|4\rangle_{B},\\
&|\phi_{7}\rangle=|0+1\rangle_{A}|2-4\rangle_{B},\\
&|\phi_{8}\rangle=|2\rangle_{A}|3-4\rangle_{B},\\
&|\phi_{9}\rangle=|0+1+2\rangle_{A}|0+1+2+3+4\rangle_{B}.\\
\end{split}
\end{eqnarray}

Then, let $|\phi_{10}\rangle=|2\rangle_{A}|0+2-3-4\rangle_{B}$, we can know $|\phi_{10}\rangle$ is orthogonal with the states in (4). Thus, the states (4) are extendible. When $m=n\geq4$, the following product states $|\phi_{2n}\rangle=|2\rangle_{A}(|0+1\rangle-2|3\rangle)_{B}$ is orthogonal to all the product states in (3). Therefore, the class of states is extendible.

\section{Conclusion}
In this paper, we construct two classes of locally indistinguishable orthogonal product states in $m\otimes n(3\leq m \leq n)$, and show the specific structure of the states. Then, we prove that Alice and Bob cannot perform a nontrivial measurement upon respective system. These results extend the phenomenon of nonlocality without entanglement. And we also hope that these results can lead to a better understanding for the phenomenon of nonlocality without entanglement. Although the smallest number of LOCC indistinguishable pure product states remains unknown, we think that the second class of states of our construction is small enough. Finally, we are interested in constructing the smallest number of pure product states which cannot be distinguished by LOCC.

\begin{acknowledgments}
This work is supported by NSFC (Grants No. 61272057 and No. 61572081) and the Beijing Higher Education Young Elite
Teacher Project (Grants No. YETP0475 and No. YETP0477).
\end{acknowledgments}

\nocite{*}

\bibliography{apssamp}

\end{document}